\documentclass[11pt, reqno]{amsart}
\usepackage{amsmath, amsthm, a4, latexsym, amssymb}

\def\@rmrk#1#2{\refstepcounter
    {#1}\@ifnextchar[{\@yrmrk{#1}{#2}}{\@xrmrk{#1}{#2}}}

\makeatletter\@addtoreset{equation}{section}\makeatother

\newfont{\bfit}{cmbxti10 scaled 2000}
\newfont{\biggi}{cmr12 scaled 2000}
\newtheorem{step}{STEP}

\newcommand{\bes}{\begin{step}}
\newcommand{\es}{\end{step}}
 
 \newcommand{\eps}{\varepsilon}

 \newcommand{\essinf}{{\rm essinf}\,}
 \newcommand{\R}{\mathbb{R}}
 \newcommand{\Z}{\mathbb{Z}}
 \newcommand{\N}{\mathbb{N}}

 \newcommand{\prob}{\mathbb{P}}

 \newcommand{\me}{\mathbb{E}}
 
 \renewcommand{\P}{\mathbb{P}}
 
 \newcommand{\one}{\1}

 \newcommand{\skrib}{{\mathcal B}}
 \newcommand{\skric}{{\mathcal C}}
 \newcommand{\skrid}{{\mathcal D}}

 \newcommand{\skrif}{{\mathcal F}}
 \newcommand{\skrig}{{\mathcal G}}

 \newcommand{\skrik}{{\mathcal K}}
 \newcommand{\skril}{{\mathcal L}}
 \newcommand{\skrim}{{\mathcal M}}

 \newcommand{\skrin}{{\mathcal N}}

 \newcommand{\skrix}{{\mathcal X}}

 \newcommand{\heap}[2]{\genfrac{}{}{0pt}{}{#1}{#2}}
 \newcommand{\sfrac}[2]{\mbox{$\frac{#1}{#2}$}}

\def\1{{\mathchoice {1\mskip-4mu\mathrm l}      
{1\mskip-4mu\mathrm l}
{1\mskip-4.5mu\mathrm l} {1\mskip-5mu\mathrm l}}}

\newcommand{\eq}{\begin{equation}}
\newcommand{\en}{\end{equation}}

\newenvironment{Proof}[1]
{\vskip0.1cm\noindent{\bf #1}{\hspace*{0.3cm}}}{\vspace{0.15cm}}

\renewcommand{\subsection}{\secdef \subsct\sbsect}
\newcommand{\subsct}[2][default]{\refstepcounter{subsection}
\vspace{0.15cm}
{\flushleft\bf \arabic{section}.\arabic{subsection}~\bf #1  }
\nopagebreak\nopagebreak}
\newcommand{\sbsect}[1]{\vspace{0.1cm}\noindent
{\bf #1}\vspace{0.1cm}}

\newtheorem{theorem}{Theorem}[section]
\newtheorem{lemma}[theorem]{Lemma}

\newtheoremstyle{thm}{1.5ex}{1.5ex}{\itshape\rmfamily}{}
{\bfseries\rmfamily}{}{2ex}{}

\newtheoremstyle{rem}{1.3ex}{1.3ex}{\rmfamily}{}
{\itshape\rmfamily}{}{1.5ex}{}
\theoremstyle{rem}

\refstepcounter{subsection}

\def\thebibliography#1{\section*{reference}
  \list%
  {\arabic{enumi}.}
    {\settowidth\labelwidth{[#1]}\leftmargin\labelwidth
    \advance\leftmargin\labelsep
    \parsep0pt\itemsep0pt
    \usecounter{enumi}}
    \def\newblock{\hskip .11em plus .33em minus .07em}
    \sloppy                   
    \sfcode`\.=1000\relax}

\begin{document}
\title[Lossy  version  of  AEP for Networked Data Structures]
{\Large Lossy  Asymptotic Equipartition property for Networked Data Structures}

\author[Kwabena Doku-Amponsah]{}

\maketitle
\thispagestyle{empty}
\vspace{-0.5cm}

\centerline{\sc{By Kwabena Doku-Amponsah}}
\renewcommand{\thefootnote}{}
\renewcommand{\thefootnote}{1}
\renewcommand{\thefootnote}{}
\footnote{\textit{Address:} Statistics Department, University of
Ghana, Box LG 115, Legon,Ghana.\,
\textit{E-mail:\,kdoku@ug.edu.gh}.}
\renewcommand{\thefootnote}{1}
\centerline{\textit{University of Ghana}}

\begin{quote}{\small }{\bf Abstract.} In  this  article  we  prove  a  Generalized  Asypmtotic  Equipartition  Property
for  Networked  Data  Structures  modelled  as  coloured  random  graphs. The  main  techniques  in  this  article  remains
 large  deviation principles  for suitably defined empirical  measures  on  coloured  random  graphs.
 We  apply  our  main  result    to  a  concrete  example  from  the  field of  Biology.

\end{quote}\vspace{0.5cm}

\section{Introduction}

Suppose   we  have  a  networked  data  structure  $x=\big\{(x(u), x(v)):\, uv\in e\big\}$ generated by a memoryless source $\skrig$ with distribution $\P^{(x)}$ is to be compressed with distortion no greater than $d \ge 0,$ using a memoryless random codebook $\hat{\skrig}$ with distribution $\P^{(y)}$. Then  the  compression performance can  be  determined by the "generalized asymptotic equipartition property" (AEP), which states that the probability of finding a $d-$ close match between $x$  and any given networked  data  structure (codeword) $y=\big\{(y(u), y(v)):\, uv\in e\big\}$, is approximately $2^{-nR(\P^{(x)},\, \P^{(y)},\, d)}.$ The rate function $R(\P^{(x)}, \P^{(y)}, d)$ can be expressed as an infimum of relative entropies. The main aim  of this  article  is  to  extend  the  results that have appeared in the recent literature  as  \cite{DA16}  and  the  reference  therein.\\

To  be  specific, in  this  article  we  develop  a Lossy  AEP  for  networked   data  structures  modelled  as  coloured  random  graphs. We   prove  process  large deviation  principle (LDP)   for  the  coloured  random  graph  conditioned to  have  a  given  empirical  colour  measure  and  empirical pair  measure,  see  Doku-Amponsah \cite{DA06}, using  similar  coupling  arguments  as  in  the article by  Boucheron et. al  \cite{BGL02}.  From  this  LDP    and   the techniques  employed  by  Dembo and Kontoyiannis~\cite{DK02} for  the  random field on  $\Z^2,$  we obtain the  proof of  the Lossy AEP for  the Networked  Data  Structures. \\

We  apply  our  Lossy  AEP  to  the  following concrete examples from biology: {\bf Metabolic network;} This is a
graph of interactions forming a part of the energy generation and
biosynthesis metabolism of the bacterium E.coli. Here, the units
represent \emph{substrates} and \emph{products}, and links represent
\emph{interactions}. See Newman \cite{New00}.\\

The  article  is  organized  as  follows. Generalized AEP for Coloured Random  Graph Model  section  contain the  main  result  of  the  paper, Theorem~\ref{AEP1}.  LDP for two-dimensional Coloured Random  Graph Model section  gives  process level LDP's, Theorem~\ref{AEP2}  and \ref{AEP4},   which  form  the  bases of  the  proof  of  the  main  result of  the  article.   Proof of Theorem~\ref{AEP1}, ~\ref{AEP2} and ~\ref{AEP4} section  provides  the  proofs  of  all  Process  Level  LDP's  for  the  paper  and  hence  the  main  result  of  the  article.

\section{Generalized  AEP for  Coloured Random  Graph  Process}
\subsection{Main Result}

Consider  two  Coloured  Random Graph  processes   $X=\big\{(X(u), X(v)):\, uv\in E\big\}$  and    
$Y=\big\{ (Y(u), Y(v)):\, uv\in E\big\}$  which  take  values  in  $\skrig=\skrig(\skrix)$  and  
$\hat{\skrig}=\hat{\skrig}(\skrix),$ resp.,  the  spaces  of  finite graphs on  $\skrix.$  
We  equip  $\skrig(\skrix)$, $\hat{\skrig}(\skrix)$  with  their  Borel $\sigma$ fields $\skrif^{(x)}$  
and  $\hat{\skrif}^{(x)}.$  Let  $\P^{(x)}$  and  $\P^{(y)}$  denote  the  probability measures  of the entire  processes $X$ and  $Y.$ 
 By $\P_{(\sigma,\pi)}^{(x)}$  and  $\P_{(\sigma,\pi)}^{(y)}$   we  denote  the   coloured  random graphs $X$ and  $Y$ conditioned 
 to have empirical colour measure  $\sigma$ and empirical pair measure~$\pi.$ See, example  \cite{DA06}. We always  assume  that  $X$  and  $Y$  are  independent  of  each  other.\\
  
 By  $\skrix$  we  denote a  finite  alphabet  and  denote   by $\skrin(\skrix)$  the  space  of  counting  measure  on   $\skrix$   
 equipped  with  the  discrete topology. By  $\skrim(\skrix)$  we  denote  the  space of  probability  measures  on  $\skrix$  equipped  with  the weak  topology  and  $\skrim_*(\skrix)$  the  space  
 of finite  measures  on  $\skrix$  equipped  with  the  weak  topology.\\   

Throughout the  rest  of  the  article  we  will  assume that  $X$  and  $Y$  are   Coloured Random  Graph processes, See \cite{Pe98}. 
For  $n\ge 1$,  let  $P_n$  denote  the  marginal  distribution  of  $X$  on  $V=\{1,2,3,...,n\}$  taking with  respect  
to  $\P_{(\sigma,\pi)}^{(y)}$ and  $Q_n^{(y)}$  denote  the  marginal distribution  $Y$  on  $V=\{1,2,3,...,n\}$  with  respect  to $\P_{(\sigma,\pi)}^{(y)}.$ \\
 
Let  $\rho:\skrix\times\skrin(\skrix)\times\skrix\times\skrin(\skrix)\to[0,\infty)$  be  
an arbitrary  non-negative  function and  define  a  sequence of  single-letter  distortion  measures  $\rho^{(n)}:\skrig\times\hat{\skrig}\to[0,\infty),$  $n\ge 1$  by
$$\rho^{(n)}(x,y)=\frac{1}{n}\sum_{v\in V}\rho\Big(\skrib_x(v),\,\skrib_y(v)\Big),$$

where $\skrib_x(v)=(x(v), L_x(v))$  and  $\skrib_y(v)=(y(v), L_y(v)).$  Given  $d\ge 0$   and  $x\in\skrig$ ,  we  denote the  distortion-ball  of  radius  $d$  by   
$$B(x,d)=\Big\{y\in\hat{\skrig}:\,\, \rho^{(n)}(x,y)\le d\Big\}.$$
For  $(\sigma,\,\pi)\in \skrim(\skrix)\times\skrim(\skrix\times\skrix),$  we write  $$\skrik_{(\sigma,\pi)}(a,l)=\sigma(a)\prod_{b\in\skrix}\frac{e^{-\pi(a,b)/\sigma(a)}[\pi(a,b)/\sigma(a)]^{\ell(b)}}{\ell(b)!},\,\mbox{for
$\ell\in\skrin(\skrix)$  } $$  and     define  the  rate  function  $I_1:\skrim[{(\skrix\times\skrin(\skrix))}^2]\to [0,\, \infty]$  by

\begin{equation}\label{AEP3}
\begin{aligned}
I_1(\nu)= \left\{ \begin{array}{ll}H\big(\nu\,\|\,\skrik_{(\sigma,\pi)}\otimes\skrik_{(\sigma,\pi)}), & 
\mbox{if $\nu$  is  consistent  and $\nu_{1,1}=\nu_{1,2}=\sigma$,}\\

 \infty & \mbox{otherwise,}

\end{array}\right.
\end{aligned}
\end{equation}

where  $$\,\skrik_{(\sigma,\pi)}\otimes\skrik_{(\sigma,\pi)}\big((a_x,a_y),(l_{x},l_{y})\big)
=\skrik_{(\sigma,\pi)}(a_x,\,l_x)\skrik_{(\sigma,\pi)}(a_y,\,l_y).$$ \\

  By  $x\,\skrid\, p$  we  mean  $x$  has distribution  $p.$  For $(\sigma,\,\pi)\in \skrim(\skrix)\times\skrim(\skrix\times\skrix),$  we  write
 $$d_{av}(\sigma,\pi)=\langle \log \langle e^{t\rho(\skrib_X,\,\skrib_Y)},\skrik_{(\sigma,\pi)}\rangle,\skrik_{(\sigma,\pi)}\rangle.$$  
 Assume $$d_{min}^{(n)}(\sigma,\pi)=\me_{P_{n}^{(x)}}\big[\essinf_{Y\,\skrid\,  Q_n^{(y)}}\rho^{(n)}(X,Y)\big]\,\mbox{ $\to  d_{min}(\sigma,\pi).$}$$  For  $n>1,$   we  write  $$R_n(P_n^{(x)},Q_n^{(y)}, d):=\inf_{V_n}\Big\{\frac{1}{n}H(V_n\,\|\,P_n^{(x)}\times Q_n^{(y)}):\,V_n\in \skrim(\skrig\times\hat{\skrig})\Big\}$$
and $$d_{min}^{\infty}(\sigma,\pi):=\inf\Big\{d\ge 0:\,\sup_{n\ge 1}R_n(P_n^{(x)}, Q_n^{(y)}, d)<\infty\Big\}.$$

\begin{theorem}\label{AEP1}
Suppose  $X$  and  $Y$  are  coloured  random  graph. Assume $\rho$ are  bounded  function.  Then,
\begin{itemize}

\item[(i)] with  $\P^{(x)}-$ probability $1,$ conditional  on  the  event $\big\{\,\Phi(\skril_{n,1})=\Phi(\skril_{n,2})=\sigma,\pi)\big\}$  the  random  variables  $\Big\{ \rho^{(n)}(x,Y)\Big\}$ satisfy  an  LDP   with  deterministic,  convex  rate-function  $$I_{\rho}(z):=\inf_{\nu}\Big\{I_1(\nu):\, \langle\rho, \,\nu\rangle=z\Big\}.$$
\item [(ii)] for  all $d\in\Big(d_{min}(\sigma,\pi),\,d_{av}(\sigma,\pi)\Big)$,  except  possibly  at $d=d_{min}^{\infty}(\sigma,\pi)$
\begin{equation}\label{AEP11}
\lim_{n\to\infty}-\frac{1}{n}\log Q_n^{(x)}\Big(B(X,D)\Big)=R\big(\P_{(\sigma,\pi)}^{(x)},\P_{(\sigma,\pi)}^{(y)},d\big)\,\,\mbox {almost  surely,}\end{equation}
where
$R(p,q,D)=\inf_{\nu}H(\nu\,\|\,p\times q).$
\end{itemize}
\end{theorem}
\subsection{Application~~\cite{DA10}}

{\bf Metabolic network.}  We consider a metabolic network  of the
energy and biosynthesis metabolism of the bacterium E.coli  modelled  as  coloured random  graph
  on  $n$ nodes partition  into   $n\sigma_n(substrate)$  block of  substrates  and  $n\sigma_n(product)$  block of
   products,  and $n\|\pi_n\|$ number of  interactions  divided into
   $n\pi_n(substrate,\, product),$  $n\pi_n(substrate,\, product),$ $ n\pi_n(substrate,\, substrate)/2,$ $ n\pi_n(product,\,product)/2$
   different interactions, respectively. Assume $\sigma_n$  converges $\sigma$  and  $\pi_n$ converges  $\pi.$   If  we  take  $\rho(s,r)=(s-r)^2$
        then, by  Theorem~\ref{AEP1}  we  have   the  distortion-rate

 \begin{equation}\label{AEP11}
\begin{aligned}
R(P,Q,D)=\left\{ \begin{array}{ll} 0, & \mbox{ if \,$D\ge 2\pi(subs,\, prod)+\pi(subs,\, subs)+\pi(prod,\,prod)+2\pi(subs,\, prod)$,}\\

 \infty & \mbox{otherwise.}

\end{array}\right.
\end{aligned}
\end{equation}

where  $subs=substracte$  and  $prod=product.$

\section{LDP  for  two-dimensional  Coloured Random Graph process}
For any $n\in\N$ we define
$$\begin{aligned}
\skrim_n(\skrix) & := \big\{ \sigma\in \skrim(\skrix) \, : \, n\sigma(a) \in \N \mbox{ for all } a\in\skrix\big\},\\
\tilde \skrim_n(\skrix \times \skrix) & := \big\{ \pi\in \tilde\skrim_*(\skrix\times\skrix)
\, : \, \sfrac n{1+\one\{a=b\}}\,\pi(a,b) \in \N  \mbox{ for all } a,b\in\skrix \big\}.
\end{aligned}$$
Throughout the proof we may assume that $\omega_n(a_x,\, a_y)>0,$   for all
$a_x,a_y\in\skrix$  and $\omega_{n,1}(a_x)=\sigma_{n}(a_x),$ $\omega_{n,2}(a_y)=\sigma_{n}(a_y)$. It is easy to see that the law of the two-dimensional
coloured graph conditioned to have empirical colour measure
$\sigma_n$ and empirical pair measure~$\pi_n$,
$$\prob_{(\sigma_n,\pi_n)}:=\prob\{ \,\cdot\,  \,|\,\Phi(\skril_{n,1})=(\omega_{n,1},\pi_n), \Phi(\skril_{n,2})=(\omega_{n,2},\pi_n)\},\\[2mm]$$
can be described in the following manner:
\begin{itemize}
\item Assign colours to the vertices by sampling without replacement from the collection
of $n$~colours, which contains any colour $(a_x,a_y)\in\skrix$ exactly $n\omega_n(a_x,\,a_y)$ times;
\item for every unordered pair $\{a,b\}$ of colours create exactly $m_n(a,b)$ edges by sampling
without replacement from the pool of possible edges connecting vertices of colour $a$ and $b$,
where
\begin{equation}\label{nabdef}
m_n(a,b):=\left\{ \begin{array}{ll} n\, \pi_n(a,\,b) & \mbox{if } a=a_x, \, b=b_x \,\mbox{and}\, a_x\not=b_x \,\\
 n\, \pi_n(a,\,b) & \mbox{if }  a=a_y, \, b=b_y \,\mbox{and}\,  a_y\not=b_y \,\\
\frac n2\, \pi_n(a,\,b) & \mbox{if } a=a_x, \, b=b_x \,\mbox{and}\, a_x=b_x \,\\
\frac n2\, \pi_n(a ,b) & \mbox{if }  a=a_y, \, b=b_y \,\mbox{and}\,  a_y=b_y
.\end{array}\right.
\end{equation}
\end{itemize}
We  define  the  process-level  empirical  measure  $\skril_n$  induced  by  $X$  and  $Y$  on  $\skrig\times\hat{\skrig}$   by
$$\skril_n(\beta_x,\beta_y)=\frac{1}{n}\sum_{v\in V}\delta_{\big(\skrib_X(v),\,\skrib_Y(v)\big)}(\beta_x,\beta_y), \, \mbox{ for $(\beta_x,\beta_y)\in\skrim[{(\skrix\times\skrix_k^*)}^2].$ }$$
Note  that  we  have
$$\begin{aligned}
\skril_n\otimes\phi^{-1}\big((x(v),y(v)),\, l_{x,y}(v)\big)&=\frac{1}{n}\sum_{v\in V}\delta_{\big(\skrib_X(v),\,\skrib_Y(v)\big)}\big(\phi^{-1}(x(v),y(v)),\, l_{x,y}(v)\big)\\
&=\frac{1}{n}\sum_{v\in V}\delta_{\big((X(v),Y(v)),\, L_{X,Y}(v)\big)}\big((x(v),y(v)),\, l_{x,y}(v)\big)\\
&:=\tilde{\skril}_n\big((x(v),y(v)),\, l_{x,y}(v)\big),
\end{aligned}$$

where  $\phi(\beta_x,\beta_y)=\big((x(v),y(v)),\, l_{x,y}(v)\big).$

The  next  Theorem  which is  the  LDP for  $\skril_n$  of  the  process  $X,Y$  is  the  main ingredient  in  the  proof  of  the  Lossy  AEP.
\begin{theorem}\label{AEP2}
The  sequence  of empirical  measures  $\skril_n$ satisfies a  large  deviation  principle  in  the  space  of  probability  measures  on   $(\skrix\times\skrin(\skrix))^2$  equipped  with  the  topology of  weak  convergence,  with  convex,  good rate-function  $I_1.$

\end{theorem}

The  proof  of  Theorem\ref{AEP2}  above   is dependent  on  the LDP  for  $\tilde{\skril}_n$  given  below:

\begin{theorem}\label{AEP4}
The  sequence  of empirical  measures  $\tilde{\skril}_n$ satisfies a  large  deviation  principle  in  the  space  of  probability  measures  on   $\skrix^2\times\skrin{(\skrix)}^2$  equipped  with  the  topology of  weak  convergence,  with  convex,  good rate-function

\begin{equation}\label{AEP5}
\begin{aligned}
I_2(\omega)= \left\{ \begin{array}{ll}H\big(\omega\,\|\,\skrik_{(\sigma,\pi)}\otimes\skrik_{(\sigma,\pi)}), & \mbox{if
$\omega$  is  consistent and   $\omega_{1,1}=\omega_{1,2}=\sigma,$}\\
\infty & \mbox{otherwise,}

\end{array}\right.
\end{aligned}
\end{equation}

where  $\,\skrik_{(\sigma,\pi)}\otimes\skrik_{(\sigma,\pi)}\big((a_x,a_y),(l_{x},l_{y})\big)
=\skrik_{(\sigma,\pi)}(a_x,\,l_x)\skrik_{(\sigma,\pi)}(a_y,\,l_y).$
\end{theorem}



We denote, for any bin
$v\in\{1,\ldots, n\}$, by $(\tilde{X}(v),\tilde{Y}(v))$ its colours, and for  $h=x,y$ by $l^v(b_h)$ the
number of balls of colour $b_h\in\skrix$ it contains. Now define the
\emph{empirical process- level  occupancy measure} of the constellation by
$$\tilde{\skril}_n^{+}(a_x,a_y,\, \ell_{x,y})
= \frac 1n \sum_{v\in V} \delta_{(\tilde{X}(v),\tilde{Y}(v), \tilde{L}_{X,Y}(v))}((a_x,a_y),\, \ell_{x,y}), \qquad \mbox{ for } (a_x,a_y,\, \ell_{x,y})
\in\skrix^2\times\skrin^2(\skrix),$$
where $\tilde{L}_{X,Y}(v)=(l^v(b_x),\,l^v(b_y), (b_x,b_y)\in\skrix\times\skrix)$ is the colour distribution
in bin $v$. In our first theorem we establish exponential
equivalence of the law of the empirical process-level
 measure $\tilde{\skril}_n$ under  $\prob_{(\sigma_n,\varpi_n)}$  the law of the  coloured   random graph conditioned to have colour
law $\sigma_n$  and edge distribution $\pi_n$.  and  the law of the empirical  process-level occupancy
measure~$\tilde{\skril}_n^{+}$ under the random allocation model
$\tilde{\prob}_{(\sigma_n,\pi_n)}$. Recall the
definition of exponential equivalence, see
\cite[Definition~4.2.10]{DZ98}.

\begin{lemma}\label{randomg.expequivalnce}
The law of $\tilde{\skril}_n^{+}$ under $\tilde{\prob}_{(\sigma_n,\pi_n)}$
is exponentially equivalent to the law of $\tilde{\skril}_n$  under
$\prob_{(\sigma_n,\pi_n)}.$
\end{lemma}

\begin{Proof}

Define the metric $d$ of total variation by
$$d(\nu,\tilde{\nu})=\sfrac{1}{2}\sum_{\big((a_x,a_y),(l_{x},l_{y})\big)\in\skrix^2\times\skrin^2(\skrix)}
|\nu\big((a_x,a_y),(l_{x},l_{y})\big)-\tilde{\nu}\big((a_x,a_y),(l_{x},l_{y})\big)|, \quad \mbox{ for
}\nu,\tilde{\nu}\in\skrim(\skrix^2\times\skrin^2(\skrix)).$$ As this
metric generates the weak topology, the proof of
Lemma~\ref{randomg.expequivalnce} is equivalent to showing that for
every $\eps>0,$
\begin{equation}\label{randomg.totalv}
\lim_{n\rightarrow\infty}\sfrac{1}{n}\log\prob\big\{d(\tilde{\skril}_n^{+} \,,\,\tilde{ \skril}_n )\ge\eps\big\}=-\infty,
\end{equation}
where $\prob$ indicates a suitable coupling between the random
allocation model and the coloured graph.

To  begin,  denote by $V(a)$ the collection of vertices (bins) which
have colour $a\in\skrix$ and observe that $$\sharp
V(a)=n\omega_n(a).$$

For $ h=x,y$  and every $a_h,b_h\in\skrix$, begin: At each step $k=1,\ldots,
m_n(a_h,b_h),$ we randomly pick two vertices $V^k_1\in V(a_h)$ and
$V^k_2\in V(b_h)$. Drop one  ball  of  colour $b_h$  in  bin  $V^k_1$
and one  ball of colour $a_h$ in $V^k_2,$  and  link $V^k_1$ to
$V^k_2$ by an edge unless $V^k_1=V^k_2$ or the two vertices are
already connected. If one of these two things happen, then we simply
choose an edge randomly from the set of all possible edges
connecting colours $a_h$ and $b_h$, which are not yet present in the
graph. This completes the construction of a graph with
$\Phi(\tilde{\skril}_{n,1})=\Phi(\tilde{\skril}_{n,2})=(\omega_n ,\,\pi_n)$ and
\begin{equation}\label{randomg.XTX}
d(\tilde{\skril}_n^{+} \,,\,\tilde{ \skril}_n )\le \sfrac{2}{n} \big(\sum_{a,b\in\skrix}B^{n}(a_x,b_x)+\sum_{a,b\in\skrix}B^{n}(a_y,b_y)\big)\, ,
\end{equation}
where $B^n(a,b)$ is the total number of steps $k\in\{1,\ldots,
m_n(a,b)\}$ at which there is disparity between the vertices
$V^k_1$, $V^k_2$ drawn and the vertices which formed the $k^{\rm
th}$ edge connecting $a$ and $b$ in the random graph construction.

Given $a,b\in\skrix$,the probability that $V^k_1=V^k_2$ or the two
vertices are already connected is equal to
$$p_{[k]}(a_h,b_h):=\sfrac{1}{m_n(a_h,b_h)}\1_{\{a_h=b_h\}}+\big(1-\sfrac{1}{m_n(a_h,b_h)}
\1_{\{a_h=b_h\}}\big)\sfrac{(k-1)}{(m_n(a_h,b_h))^2}.$$ $B^{n}(a_h,b_h)$ is a
sum of independent Bernoulli random variables
$X_1^{(h)},\,...,\,X_{n\varpi_n(a_h,b_h)/2}^{(h)}$ with `success' probabilities
equal to $p_{[1]}(a_h,b_h), \ldots, p_{[n\varpi_n(a_h,b_h)/2]}(a_h,b_h)$. Note
that $\me[X_k]= p_{[k]}(a_h,b_h)$  and
$$Var[X_k^{(h)}]=p_{[k]}(a_h,b_h)(1-p_{[k]}(a_h,b_h)).$$  Now,  we   have   $$\me B^{n}(a_h,b_h)= \sum_{k=1}^{n(a_h,b_h)} p_{[k]}(a_h,b_h)=\1_{\{a_h=b_h\}}
+\big(1-\1_{\{a_h=b_h\}}\sfrac{1}{m_n(a_h,b_h)}\big)\big(1-\sfrac{1}{m_n(a_h,b_h)}\big)\le
1+\1_{\{a_h=b_h\}}.$$

We  write
$$\sigma_n^2(a_h,b_h):=\sfrac{1}{m_n(a_h,b_h)}\sum_{k=1}^{m_n(a_h,b_h)}Var[X_k^{(h)}]$$
and  observe  that   $$\lim_{n\to \infty}\me(B^n(a_h,b_h))=\lim_{n\to
\infty}Var(B^n(a_h,b_h))=\lim_{n\to
\infty}m_n(a_h,b_h)\sigma_n^2(a,b)=\1_{\{a_h=b_h\}}+1.$$

We  Define  $e(t)=(1+t)\log(1+t)-t,$ for  $t\ge 0$   and use
Bennett's inequality, see \cite{Be62}, to  obtain, for  sufficiently
large $n$
$$\P\Big\{ \sfrac{1}{n}\sum_{h=x,y} B^{n}(a_h,b_h)\ge\sfrac{ \sum_{h=x,y}\1_{\{a_h=b_h\}}+1}{n}+\delta_{1}\Big\}
\le \exp\Big[-\sum_{h=x,y}m_n(a_h,b_h)\sigma_n^2(a_h,b_h)e(\sfrac{n\delta_{1}}{\sum_{h=x,y}m_n(a_h,b_h)\sigma_n^2(a_h,b_h)})\Big],$$
for any $\delta_1>0.$ Let $\eps\ge 0 $ and  choose
$\delta_1=\sfrac{\eps}{2m^2}.$ Suppose that we have
$B^n(a_h,b_h)\le\delta_1,$  for  $h=x,y$. Then, by~\eqref{randomg.XTX},
$$d(\tilde{\skril},{\nu}_n)\le 2\delta_1 m^2=\eps.$$ Hence,
$$\begin{aligned}
\prob\big\{ d(\tilde{\skril},\tilde{\skril}^{+})> \eps \big\}  \le \max_{h=x,y}\sum_{a_h,b_h\in\skrix}
\prob\big\{ B^n(a_h,b_h)&\ge n\delta_1 \big\}\\
&\le
m^2\max_{h=x,y}\sup_{a_h,b_h\in\skrix}\prob\big\{ B^n(a_h,b_h)\ge \1_{\{a_h=b_h\}}+1+
(n\delta_1)/2
\big\}\\
 & \le m^2\max_{h=x,y}\sup_{a,b\in\skrix} exp\Big[-m_n(a_h,b_h)\sigma_n^2(a_h,b_h)e(\sfrac{n\delta_{1}}{m_n(a_h,b_h)\sigma_n^2(a_h,b_h)})\Big] .
\end{aligned}$$

Let   $0\le \delta_2\le 1$. The,  for  sufficiently  large $n$ we
 have
\begin{equation}\begin{aligned}\label{Equ.coupling}
\frac{1}{n}&
\log\P\Big\{d(\tilde{\skril},\tilde{\skril}^{+}) > \eps \Big\}\le-(1-\delta_{2})e(\sfrac{n\delta_1}{2(1+\delta_{2})})\\
&=-(\1_{\{a=b\}}+1-\delta_{2})\Big[(\sfrac{1}{n}+\sfrac{\delta_1}{2(\1_{\{a=b\}}+1+\delta_{2})})\log(1+\sfrac{n\delta_1}{2(\1_{\{a=b\}}+1+\delta_{2})})-\sfrac{\delta_1}{2(\1_{\{a=b\}}+1+\delta_{2})}\Big].
\end{aligned}\end{equation}

This completes the proof of the lemma.
\end{Proof}


\section{Proof  of  Theorem~\ref{AEP1}, \ref{AEP2} and~\ref{AEP4}}

\subsection{Proof  of  Theorem~\ref{AEP4}.}
We  write  $\vartheta_{2}^{(n)}:=\vartheta_{2}^{(n)}(\varpi_n,\nu_n)$,  $\vartheta_{1}^{(n)}:=\vartheta_{1}^{(n)}(\varpi_n,\nu_n)$  and  state  the  following  Lemmma.
Denote  by  $\Sigma^{(n)}(\sigma_n,\pi_n)$  the  space  of  all  empirical neighbourhood  measures  with   
empirical  colour  measure $\sigma_n$  and  empirical  pair  measure  $\pi_n.$

\begin{lemma}[Doku-Amponsah, 2014]\label{randomg.LDprobm}
For any   process level empirical  measure,   $\nu_n$  with   $\nu_{n,1},\nu_{n,2}\in\Sigma^{(n)}(\sigma_n,\pi_n),$
\begin{equation}\label{AEP71}
\begin{aligned}
e^{-n(H(\nu_{n,1}\,\|\,\skrik_{(\sigma_n,\pi_n)})+H(\nu_{n,2}\,\|\,\skrik_{(\sigma_n,\pi_n)})+\vartheta_{1}^{(n)}} &\le\tilde{\prob}_{(\sigma_n,\pi_n)}(\tilde{\skril}_n^{+}=\nu_n)\\
&\le|\Sigma^{(n)}(\sigma_n,\pi_n)|^{-2}
e^{-n(H(\nu_{n,1}\,\|\,\skrik_{(\sigma_n,\pi_n)})+H(\nu_{n,2}\,\|\,\skrik_{(\sigma_n,\pi_n)})+\vartheta_{2}^{(n)}},
\end{aligned}
\end{equation}
where    $\skrik_{(\sigma_n,\pi_n)}(a_h,l_h)=\sigma_n(a_h)\skrik_{\pi_n}\{l_h \,|\,a_h\}$  and
$$\skrik_{\pi_n}\{l_h \,|\,a_h\}=\prod_{b_h\in\skrix}\frac{e^{-\pi_n(a_h,b_h)/\sigma_n(a_h)}[\pi_n(a_h,b_h)/\sigma_n(a_h)]^{\ell(b_h)}}{\ell(b_h)!},\,\mbox{for
$\ell_h\in\skrin(\skrix)$ and \,  $h=x,y.$ } $$

$$\lim_{n\to \infty}\vartheta_{2}^{(n)}=\lim_{n\to \infty}\vartheta_{1}^{(n)}=0.$$

\end{lemma}
\begin{proof}
Note,  by construction   For any   process level empirical  measure,   $\nu_n$  with   $\nu_{n,1},\nu_{n,2}\in\Sigma^{(n)}(\sigma_n,\pi_n),$ we  have
\begin{align}
&\tilde{\prob}_{(\sigma_n,\pi_n)}(\tilde{\skril}_n^{+}=\nu_n)=\tilde{\prob}\big\{\tilde{\skril}_n^{+}=\nu_n\,\big|\,\Phi(\tilde{\skril}_{n,1}^{+})=\Phi(\tilde{\skril}_{n,2}^{+})=(\sigma_n,\pi_n)\big\}\\
&=\prod_{h=x,y}\prod_{a_h\in\skrix}\Big(\heap{n\sigma_n(a_h))}{n\nu_{n,u(h)}(a_h,\ell_h),\,\ell_h\in\skrin(\skrix)}\Big)\prod_{a_h,b_h\in\skrix}\Big(\heap{n\pi_n(a_h,b_h
)}{\ell_{a_h}^{(j)}(b_h),\,j=1,...,n\omega_n(a_h)}\Big)\Big(\frac{1}{n\sigma_n(a_h)}\Big)^{n\pi_n(a_h,b_h)},\label{Tclass}
\end{align}
while $\tilde{\prob}_{(\sigma_n,\pi_n)}(\tilde{\skril}_{n}^{+})=0$ when $\Phi(\tilde{\skril}_{n,1}^{+})\not=(\sigma_n,\pi_n)$  or
$\Phi(\tilde{\skril}_{n,2}^{+})\not=(\sigma_n,\pi_n)$ by convention. Therefore,  by  similar  combinatoric  computations  as   in the  proof  of   \cite[Lemma~0.6]{DA14}  and  the  Sterling's  formula  see,  \cite{Fe67}   we  have   \ref{AEP71}.
\end{proof}

The  proof  of   Theorem~\ref{AEP4}  follows  from  Lemma~\ref{randomg.LDprobm}   and  similar   arguments  as  \cite[Page~13]{DA14}.

\subsection{Proof  of  Theorem~\ref{AEP2}.}
Let  $\Gamma\in\skrim[(\skrix \times \skrin(\skrix))^2]$  and  write  $ \Gamma_{\phi}=\big\{ \omega\otimes\phi^{-1}:\, \omega\in \Gamma\big\}.$
Note  that  if  $A$  is  closed (open)  then  $\Gamma_{\phi}$  is  closed (open) since  $\phi$  is  linear.  Now  suppose    $F$  is  closed  subset  of
$\skrim[(\skrix \times \skrin(\skrix))^2]$   then   by  Theorem~\ref{AEP4} we  have
$$\begin{aligned}
-\inf_{\omega\in F} I_2(\omega\otimes \phi^{-1})=-\inf_{\nu\in F_{\phi}}I_2(\nu)&\le\liminf_{n\to\infty}\sfrac{1}{n}\log\P\big \{\tilde{\skril}_n\in F_{\phi}\big\}\\
&\le \lim_{n\to\infty}\sfrac{1}{n}\log\P\big\{ \skril_n\in F\big\}\le\limsup_{n\to\infty}\sfrac{1}{n}\log\P\big \{\tilde{\skril}_n\in F_{\phi}\big\}\\
&\le -\inf_{\nu\in F_{\phi}}I_2(\nu)=-\inf_{\omega\in F} I_2(\omega\otimes \phi^{-1}).
\end{aligned}$$
We obtain  the  form  of  the  rate  function  in  Theorem~\ref{AEP2}  if  we  solve  the  optimization  problem
$$\inf\Big\{I_2(\nu):\,\omega\otimes \phi^{-1}=\nu\Big\}=I_1(\omega).$$
\subsection{Proof  of  Theorem~~\ref{AEP1}}

We  write  $\skrim:=\skrim[(\skrix \times \skrin(\skrix))^2]$  and  define  the  set  $\skric^{\eps}$  by
$$\displaystyle\skric^{\eps}(\sigma,\,\pi)=\Big\{ \nu\in\skrim\colon
\sup_{\beta_x,\beta_y\in\skrix\times\skrin(\skrix)} |\nu(\beta_x,\,\beta_y) - \skrik_{(\sigma,\pi)}\otimes\skrik_{(\sigma,\pi)}(\beta_x,\,\beta_y)| \ge \eps\Big\}.$$

\begin{lemma}\label{WLLN}  Suppose  the  sequence of  measures  $(\sigma_n,\pi_n)$  converges to  the  pair  of  measures  $(\sigma,\pi).$
For any $\eps>0$ we have  $\lim_{n\to\infty} \P_{(\sigma_n,\pi_n)}\big(\skric^{\eps}\big)=0.$

 \end{lemma}
\begin{proof}

Observe  that  $\skric^{\eps}$  defined  above  is  a closed  subset  of  $\skrim$  and  so  by  Theorem~\ref{AEP2}  we  have  that

\begin{equation}\label{AEP7}
\limsup_{n\to\infty}\frac{1}{n}\log\P_{(\sigma_n,\pi_n)}\big(\skric^{\eps}\big)\leq  -\inf_{\nu\in\skric^{\eps}}I_1(\nu).
\end{equation}

We use  proof  by contradiction  to  show that the right hand side of \eqref{AEP7} is negative.Suppose that there exists sequence  $\nu_n$ in $\skric^{\eps}$ such that
$I_1(\nu_n)\downarrow 0.$ Then, there is a limit point $\nu\in F_1$ with $I(\nu)=0.$ Note $I$ is a good rate function and its level sets are compact,
and the mapping $\nu\mapsto I(\nu)$) lower semi-continuity. Now $I_1(\nu)=0$  implies $\nu(\beta_x,\,\beta_y)=\skrik_{(\sigma,\pi)}\otimes\skrik_{(\sigma,\pi)}(\beta_x,\,\beta_y),$  for  all   $\beta_x,\beta_y\in\skrix\times\skrin(\skrix) $ which contradicts  $\nu\in\skric^{\eps}$.

\end{proof}
(i)    Notice  $\displaystyle\rho^{(n)}(X,Y)=\langle\rho, \,\skril_n\rangle$  and  if  $\Gamma$  is open (closed)   subset  of  $\skrim$ then  $$ \Gamma_{\rho}:=\big\{ \nu: \langle\rho, \,\nu\rangle\in \Gamma\big\}$$ is  also open (closed) set  since  $\rho$  is  bounded function.

$$\begin{aligned}
-\inf_{z\in In(\Gamma)}I_{\rho}(z)&=-\inf_{\nu\in\ln(\Gamma_{\rho})}I_1(\nu)\\
&\le\liminf_{n\to\infty}\sfrac{1}{n}\log\P \Big\{\rho^{(n)}(X,Y)\in\Gamma\big |X=x,\,\Phi(\skril_{n,1})=\Phi(\skril_{n,2})=(\sigma_n,\pi_n)\Big\}\\
 &\le\lim_{n\to\infty}\sfrac{1}{n}\log\P \Big\{\rho^{(n)}(X,Y)\in\Gamma\big |X=x,\,\Phi(\skril_{n,1})=\Phi(\skril_{n,2})=(\sigma_n,\pi_n)\Big\}\\ 
 &\le\limsup_{n\to\infty}\sfrac{1}{n}\log\P\Big\{\rho^{(n)}(X,Y)\in\Gamma\big |X=x,\,\Phi(\skril_{n,1})=\Phi(\skril_{n,2})=(\sigma_n,\pi_n)\Big\}\\
 &\le -\inf_{\nu\in cl(\Gamma_{\rho})}I_1(\nu)=-\inf_{z\in cl(\Gamma)}I_{\rho}(z).
 \end{aligned}$$

(ii)  Observe  that  $\rho$ are  bounded, therefore  by  Varadhan's  Lemma and  convex duality, we  have
$$R( \P^x, \P^y, d)=\sup_{t\in\R}[td-\Lambda_{\infty}(t)]=\Lambda_{\infty}^{*}(d)$$
where
$$\Lambda_{\infty}^{*}(t):=\lim_{n\to\infty}\sfrac{1}{n}\log \int e^{nt\Big\langle\rho, \,\skril_n\Big\rangle}dQ_n(y)$$
exits  for  $\P$ almost  everywhere  $x.$  Using  bounded  convergence,  we  can  show  that  $$\Lambda_{\infty}(t)=\lim_{n\to\infty}\Lambda_n(t)
:=\lim_{n\to\infty}\sfrac{1}{n}\int \Big[\log \int e^{nt\Big\langle\rho, \,\skril_n\Big\rangle}dQ_n(y)\Big]dP_n(x).$$
Using   Lemma~\ref{AEP7},  by  boundedness of  $\rho$ we  have  that
$$\sfrac{1}{n}\Lambda(nt)=\frac{1}{n}\sum_{j=1}^{n}\log\me_{Q_n}\big(e^{t\rho(\skrib_x(j),\skrib_y(j)}\big)\to\langle \log \langle e^{t\rho(\skrib_X,\skrib_Y)},\skrik_{(\sigma,\pi)}\rangle,\skrik_{(\sigma,\pi)}\rangle=d_{av}(\sigma,\pi).$$
Also let  $$D_{min}^{(n)}:=\lim_{t\downarrow-\infty}\sfrac{\Lambda_n(t)}{t}$$
so  that  $\Lambda_n^{*}(d)=\infty$ for  $d< d_{min}^{(n)}$,  while   $\Lambda_n^{*}(D)<\infty$ for  $d>d_{min}^{(n)}.$  Observe  that  for  $n<\infty$  we  have   $D_{min}^{(n)}(d)=\me_{P_n}\big[\essinf_{Y\,\skrid\, Q_n}\rho^{(n)}(X,Y)\big],$  which  converges to  $d_{min}.$
 Using  similar  arguments  as  \cite[Proposition~2]{DK02}  we  obtain  $$ R_n(P_n,Q_n,d)=\sup_{t\in\R}\big(td-\Lambda_n(t)\big):=\Lambda_{n}^{*}(d)$$

Now  we   observe  from \cite[Page 41]{DK02}  that  the  converge  of  $\Lambda_{n}^{*}(\cdot)\to\Lambda_{\infty}(\cdot)$  is  uniform on  compact  subsets  of  $\R.$  Moreover, $\Lambda_{n}$    convex,  continuous  functions    converge  informally  to  $\Lambda_{\infty}$  and  hence  we  can invoke  \cite[Theorem 5]{Sce48}  to  obtain

 $$\Lambda_{n}^{*}(d)=\lim_{\delta\to 0}\limsup_{n\to\infty}\inf_{|\hat{d}-d|<\delta}\Lambda_{n}^{*}(\hat{d}).$$

  Using similar  arguments as \cite[Page 41]{DK02} in  the lines  after  equation  (64)  we  have  \eqref{AEP11}  which  completes  the  proof.

{\bf \Large Conflict  of  Interest}

The  author  declares  that  he has  no  conflict  of  interest.\\

{\bf \Large  Acknowledgement}

This  extension  has  been  mentioned  in the author's  PhD Thesis  at  University  of  Bath.



\end{document}